\documentclass[a4paper,UKenglish, cleveref, autoref, numberwithinsect, thm-restate]{lipics-v2019}

\pdfoutput=1

\usepackage{caption}
\usepackage[noend,linesnumbered]{algorithm2e}
\usepackage{graphicx}
\usepackage{amsmath,amsthm,amssymb,hyperref,color}

\title{Weakly Submodular Function Maximization Using Local Submodularity Ratio} 

\titlerunning{Weakly Submodular Function Maximization Using Dynamic Submodularity Ratio} 

\author{Richard Santiago}{ETH Zurich, Switzerland}{rtorres@ethz.ch}{}{}

\author{Yuichi Yoshida}{National Institute of Informatics, Japan}{yyoshida@nii.ac.jp}{}{}

\authorrunning{Richard\,Santiago and Yuichi\,Yoshida} 

\Copyright{Richard Santiago and Yuichi Yoshida} 



\acknowledgements{Part of this work was conducted during the first author’s visit to the National Institute of Informatics in Tokyo. The second authors was supported by JSPS KAKENHI Grant Number 18H05291.}

\nolinenumbers 

\hideLIPIcs  

\EventEditors{Yixin Cao, Siu-Wing Cheng, and Minming Li}
\EventNoEds{3}
\EventLongTitle{31st International Symposium on Algorithms and Computation (ISAAC 2020)}
\EventShortTitle{ISAAC 2020}
\EventAcronym{ISAAC}
\EventYear{2020}
\EventDate{December 14--18, 2020}
\EventLocation{Hong Kong, China (Virtual Conference)}
\EventLogo{}
\SeriesVolume{181}
\ArticleNo{5}

\newtheorem{observation}[theorem]{Observation}

\newcommand{\R}{\mathbb{R}}
\newcommand{\E}{\mathbb{E}}
\renewcommand{\P}{\mathbb{P}}
\newcommand{\greedy}{\textsf{RandomizedGreedy }}

\DeclareMathOperator*{\argmax}{arg\,max}

\begin{document}
\maketitle

\begin{abstract}
  Weak submodularity is a natural relaxation of the diminishing return property, which is equivalent to submodularity.
  Weak submodularity has been used to show that many (monotone) functions that arise in practice can be efficiently maximized
  with provable guarantees.
  In this work we introduce two natural generalizations of weak submodularity for non-monotone functions. We show that
  an efficient randomized greedy algorithm has provable approximation guarantees for maximizing these functions subject to a cardinality constraint.
  We then provide a more refined analysis that takes into account that the weak submodularity parameter may change (sometimes improving) throughout the execution of the algorithm. This leads to improved approximation guarantees in some settings.
  We provide applications of our results for monotone and non-monotone maximization problems.
 \end{abstract}

\section{Introduction}
Submodularity is a property of set functions equivalent to the notion of diminishing returns. More formally, we say that a set function $f:2^E \to \R$ is \emph{submodular} if for any two sets $A\subseteq B \subseteq E$ and an element $e \notin B$, the corresponding marginal gains satisfy $f(A \cup \{e\}) -f(A) \geq f(B \cup \{e\}) -f(B)$. Submodularity has found a wide range of connections and applications to different computer science areas in recent years.

However, many applications in practice does not satisfy the diminishing returns property, but rather a weaker version of it. This has motivated several lines of work exploring different ways to relax the submodularity property 
\cite{das2011submodular,feige2013welfare,chen2018capturing,horel2016maximization,feige2015unifying,ghadiri2019beyond,ghadiri2020parameterized}.
One such relaxation that has received a lot of attention from the machine learning community is the notion of weak submodularity (we postpone the formal definition to Section~\ref{sec:definitions}), originally introduced by Das and Kempe~\cite{das2011submodular}.
They provided applications to the feature selection and the dictionary selection problems, and showed that the standard greedy algorithm achieves a $(1-e^{-\gamma})$-approximation for the monotone maximization problem subject to a cardinality constraint.
Here the parameter $\gamma \in [0,1]$ is called the submodularity ratio, and it measures how ``close'' the function is to being submodular.
Weak submodularity has found applications in areas such as linear and nonlinear sparse regression~\cite{elenberg2017streaming,khanna2017scalable}, high-dimensional subset selection~\cite{elenberg2018restricted}, interpretability of black-box neural network classifiers~\cite{elenberg2017streaming}, video summarization, splice site detection, and black-box interpretation of images~\cite{chen2018weakly}.

In subsequent work, Das and Kempe~\cite{das2018approximate}
left as an open question whether some of these theoretical guarantees can be extended to non-monotone objectives.
As their definition of weak submodularity is targeted at monotone functions, they raise the question of whether there is a more general definition  that retains some of the positive results of their work, while also yielding an analogue to non-monotone objectives.

One main goal of this work is to answer that question. We believe this is interesting for both theoretical and practical purposes, given that non-monotone submodular objectives have found a wide range of applications in computer science. Some of these include document summarization~\cite{lin2010multi,lin2011class}, MAP inference for determinantal point processes~\cite{gillenwater2012near}, personalized data summarization~\cite{mirzasoleiman2016fast}, nonparametric learning~\cite{zoubin2013scaling},  image summarization~\cite{tschiatschek2014learning}, and removing redundant elements from DNA sequencing~\cite{libbrecht2018choosing}. Hence, it seems natural to study how the approximation guarantees for non-monotone submodular maximization degrade
in terms of the submodularity ratio.

In this work we introduce a natural generalization of weak submodularity to the non-monotone setting.
We then show that a fast and simple randomized greedy algorithm retains some of the good theoretical guarantees available for (non-monotone) submodular objectives.
In addition, for monotone weakly submodular functions, this algorithm retains the approximation guarantee of $1-e^{-\gamma}$ given in~\cite{das2011submodular}.

A second main contribution of our work is to provide a more refined analysis that takes into account that the submodularity ratio parameter may change (some times improving) throughout the execution of the algorithm.
We provide several applications where this more refined bound leads to improved approximation guarantees, for both monotone and non-monotone maximization problems.

The rest of this section is organized as follows. In Section~\ref{sec:definitions} we extend weak submodularity to the non-monotone setting. In Section~\ref{sec:gamma_A,B} we discuss the notion of local submodularity ratio. We discuss several examples and applications in Section~\ref{sec:examples}. Our main contributions are presented in Section~\ref{sec:contributions}.  Additional related work regarding weak submodularity and non-monotone submodular maximization is discussed in Section~\ref{sec:related-work}.

\subsection{Weak submodularity and non-monotonicity}\label{sec:definitions}

Throughout this paper we use $f_A(B)$ to denote the marginal gain of adding the set $B$ to $A$, that is $f(A \cup B) - f(A)$.
A non-negative monotone set function $f:2^E \to \R_+$ is \emph{$\gamma$-weakly submodular} for some parameter $0 \leq \gamma \leq 1$, if for any pair of disjoint sets $A,B \subseteq E$, it satisfies
$
\sum_{e \in B} f_A (e) \geq \gamma \cdot f_A (B).
$
We note that this is the definition used in~\cite{bian2017guarantees,chen2018weakly,elenberg2017streaming}, which is slightly adapted from the original definition given in~\cite{das2011submodular, das2018approximate}. The parameter $\gamma$ is called the \emph{submodularity ratio}.

When $f$ is monotone, it is clear that for any value of $\gamma \in [0,1]$ the above class contains monotone submodular functions. However, for non-monotone objectives the marginal gains can be negative, and in this case we have $\gamma f_A (B) \geq f_A (B)$ whenever $f_A (B) \leq 0$, leading to a stronger condition than diminishing returns. This motivates us to introduce the following two classes of non-monotone non-submodular functions.



\begin{definition}[pseudo and weak submodularity]

	Given a scalar $0< \gamma\leq1$, we say that a set function $f:2^E \to \R_+$ is:
	\vspace*{0.05cm}
	\begin{enumerate}
		\item $\gamma$-pseudo submodular if $\sum_{e \in B} f_A (e) \geq \gamma f_A (B)$ for any pair of disjoint sets $A,B \subseteq E$.
		\vspace*{0.05cm}
		\item $\gamma$-weakly submodular if $\sum_{e \in B} f_A (e) \geq \min\{ \gamma f_A (B), \frac{1}{\gamma} f_A (B)\}$ for any $A,B \subseteq E$ disjoint.
	\end{enumerate}
\end{definition}

We first note that for monotone functions, the above two definitions are equivalent to the notion of $\gamma$-weakly submodularity from previous works~\cite{bian2017guarantees,chen2018weakly,elenberg2017streaming}. This follows immediately from the fact that monotone functions satisfy $f_A(B) \geq 0$ for all $A,B \subseteq E$, and hence $\min\{ \gamma f_A (B), \frac{1}{\gamma} f_A (B)\} = \gamma f_A (B)$.

For any value $\gamma \in (0,1]$ the above definition of $\gamma$-weakly submodularity leads to a weaker notion of diminishing returns (i.e., it contains non-monotone submodular functions). Indeed, if $f_A (B) \geq 0$ we have $\sum_{e \in B} f_A (e) \geq \gamma f_A (B)$, while if  $f_A (B) < 0$ we have $\sum_{e \in B} f_A (e) \geq \frac{1}{\gamma} f_A (B)$.  On the other hand, while the class of $\gamma$-pseudo submodular functions does not properly contain non-monotone submodular functions, it does contain functions that are not necessarily submodular. We show this 
in Figure~\ref{fig:function-hierarchy}.


\begin{figure}
	\centering
	\includegraphics[scale=0.8]{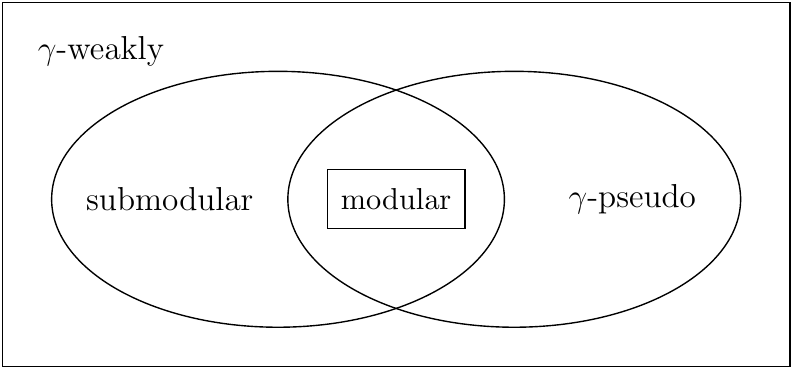}
	\caption{Hierarchy of the different function classes.}\label{fig:function-hierarchy}
\end{figure}

\subsection{Local submodularity ratio}\label{sec:gamma_A,B}

The submodularity ratio $\gamma$ is in general a very pessimistic bound for most applications. This is due to the fact that $\gamma$ is defined as a \emph{global} bound, in the sense that it must hold for any pair of disjoint sets $A,B \subseteq E$. Or at least for any pair of sets that are relevant to the execution of the algorithm, e.g., sets of cardinality at most $k$. We next discuss a natural way to ``refine'' this bound.

Given a function $f:2^E \to \R_+$ and any pair of disjoint sets $A,B \subseteq E$, in this work we denote by $\gamma^f_{A,B}$ any non-negative scalar satisfying $\sum_{e \in B} f_A (e) \geq \gamma^f_{A,B} \cdot f_A (B)$. When it is clear from the context we usually simplify the notation to $\gamma_{A,B}$ instead of $\gamma^f_{A,B}$. One of our contributions is showing how using these local bounds can be beneficial in some settings. In particular we discuss several natural classes of functions for which (i) one can compute explicit bounds for the value $\gamma_{A,B}$ (see Section~\ref{sec:examples}), and (ii) using the local bounds $\gamma_{A,B}$ (instead of $\gamma$) leads to significantly better theoretical guarantees (we discuss this in more detail in Section~\ref{sec:contributions}).
We believe this is interesting for both theoretical and practical applications.

\subsection{Examples and applications}\label{sec:examples}

In this section we present several classes of functions for which the parameter $\gamma_{A,B}$ can be bounded explicitly, and discuss  applications arising from these results. 
We postpone the proofs to Appendix~\ref{sec:appendix-examples-proof}.

Our first example is the so-called metric diversity function (also known as remote clique). Here we are given a metric (i.e., a distance that satisfies the triangle inequality)  $d:E \times E \to \R_+$ over a finite set $E$, where $d(u,v)$ measures the dissimilarity between two elements $u$ and $v$. One then defines a set function $f(S)= \frac{1}{2} \sum_{u\neq v \in S} d(u,v)$ that measures the diversity inside the set $S$. The problem $\max d(S): |S| \leq k$ of finding a diverse subset has been studied in the operations research community~\cite{hassin1997approximation, ravi1994heuristic, birnbaum2009improved}, and has found applications in other areas~\cite{agrawal2009diversifying, drosou2010search}.

\begin{restatable}{example}{metrics}\label{ex:diversity}
	Given a metric  $d:E \times E \to \R_+$, consider the function given by $f(S)=\sum_{ \{u,v\} \subseteq S} d(u,v)$, which is monotone and supermodular. Then, we have
	$
	\gamma_{A,B} \geq \frac{a}{a+b-1}
	$
	for any two disjoint sets $A,B \subseteq E$, where $a = |A|$ and $b = |B|$.

\end{restatable}

The works of~\cite{borodin2014weak,borodin2015proportionally} introduced the notion of proportionally submodular functions\footnote{They called them weakly submodular at first, and changed the name in subsequent work.}.
A set function $f: 2^E \to \mathbb{R}_+$ is \emph{proportionally submodular} if
$
	|S|f(T) + |T|f(S) \geq |S \cap T|f(S \cup T) + |S \cup T|f(S \cap T)
$
for every $S,T \subseteq E$. In the monotone setting, this class properly contains monotone submodular functions. In addition, this class also contains some non-submodular objectives such as the (supermodular) metric diversity function discussed in Example~\ref{ex:diversity}. Since these functions are closed under addition, the sum of a monotone submodular function and a metric diversity function is proportionally submodular. Our next result bounds the parameter $\gamma_{A,B}$ for this class, in both the monotone and non-monotone settings.

\begin{example}\label{ex:proportionally-submod}
	A non-negative proportionally submodular function $f: 2^E \to \mathbb{R}_+$ has
	$
	\gamma_{A,B} \geq \frac{3 a (1 + a)}{3 a^2 + 3 a b + b^2 - 1}
	$
	for any two disjoint sets $A,B \subseteq E$,
	where $a=|A|$ and $b=|B|$.
\end{example}

The above result leads to interesting applications. First it allows to improve over the current best approximation for maximizing a monotone proportionally submodular function subject to a cardinality constraint. In addition, combining this with other results from this work, we can also get improved approximations for the product $f\cdot g$ of a monotone submodular function $f$ and a monotone proportionally submodular function $g$. We discuss this in more detail in Section~\ref{sec:applications}.

\begin{restatable}{example}{product}\label{ex:f(S)g(S)}
	Let $f,g:2^E \to \mathbb{R}_+$ be two monotone set functions with parameters $\gamma^f_{A,B}$ and $\gamma^g_{
		A,B}$ respectively. Then the product function $h(S):=f(S) \cdot g(S)$ is also non-negative and monotone, with parameter
	\[
	\gamma_{A,B} \geq
	\begin{cases}
	\frac{f(A)}{f(A \cup B)} \gamma^g_{A,B} & \text{if } \gamma^f_{A,B} \geq \gamma^g_{A,B} \\
	\frac{g(A)}{g(A \cup B)} \gamma^f_{A,B} & \text{if } \gamma^g_{A,B} \geq \gamma^f_{A,B},
	\end{cases}
	\]
	for any two disjoint sets $A,B \subseteq E$. In particular, if $f$ and $g$ have global parameters $\gamma^f$ and $\gamma^g$ respectively, such that $\gamma^f \geq \gamma^g$, then the product function $h$ has parameter $\gamma_{A,B} \geq  \gamma^g \cdot \max \{ \frac{f(A)}{f(A \cup B)} , \frac{g(A)}{g(A \cup B)} \}$.
\end{restatable}

Using that submodular functions satisfy $\gamma_{A,B}\geq 1$, we can combine the above result with Examples~\ref{ex:diversity} and~\ref{ex:proportionally-submod} to get the following.

\begin{example}\label{ex:f(S)g(S)-apps}
	Let $f,g : 2^E \to \mathbb{R}_+$ be two monotone functions, and let $h(S):= f(S) \cdot g(S)$ be the product function with parameter $\gamma_{A,B}$. Then we have the following.
	\begin{enumerate}[(a)]
		\item If $f$ and $g$ are submodular then $\gamma_{A,B} \geq \max \{\frac{f(A)}{f(A \cup B)}, \frac{g(A)}{g(A \cup B)}\}$.

		\item If $f$ is submodular and $g$ is the metric diversity function from Example~\ref{ex:diversity}, then $\gamma_{A,B} \geq \frac{f(A)}{f(A \cup B)} \cdot \frac{a}{a+b-1}$, where $a=|A|$ and $b=|B|$.

		\item If $f$ is submodular and $g$ is proportionally submodular then $\gamma_{A,B} \geq \frac{f(A)}{f(A \cup B)} \cdot \frac{3 a (1 + a)}{3 a^2 + 3 a b + b^2 - 1}$, where $a=|A|$ and $b=|B|$.
	\end{enumerate}

\end{example}

By taking a non-monotone submodular function $f$, and either multiplying it or dividing by the cardinality function, we obtain a new function that is no longer submodular. The next example bounds the parameter $\gamma_{A,B}$ for these functions.

\begin{restatable}{example}{prodcardisubmod}\label{ex:|S|f(S)}
	Let $f:2^E \to \mathbb{R}_+$ be a submodular function. Then for any two disjoint sets $A,B \subseteq E$ with $|A|=a$ and $|B|=b$ we have the following.
	\begin{enumerate}[(a)]
		\item The function $g(S):=|S| \cdot f(S)$ satisfies $\gamma_{A,B} \geq \frac{a+1}{a+b}$.

		\item The function $g(S):=\frac{f(S)}{|S|}$ has $\gamma_{A,B} \leq \frac{a+b}{a+1}$.
	\end{enumerate}

\end{restatable}

We next discuss the behavior of the parameter $\gamma_{A,B}$ under summation, 
and how this result allows us to generalize some of the bounds previously discussed in this section.

\begin{restatable}{proposition}{sumproperty}\label{prop:sum-property}
	Let $f,g:2^E \to \mathbb{R}_+$ be two set functions with parameters $\gamma^f_{A,B}$ and $\gamma^g_{A,B}$ respectively. We have the following.
	\begin{enumerate}[(a)]
		\item If $f$ and $g$ are both monotone, then $f+g$ is also monotone with parameter $ \gamma_{A,B} \geq \min\{\gamma^f_{A,B}, \gamma^g_{A,B} \}$. In particular, if $0 \leq \gamma^g_{A,B} \leq \gamma^f_{A,B}$ holds for all pairs of disjoint sets $A$ and $B$, then $f+g$ has parameter $\gamma_{A,B} \geq \gamma^g_{A,B}$.

		\item If $f$ is monotone and $g$ is non-monotone, and $0 \leq \gamma^g_{A,B} \leq \gamma^f_{A,B}$ holds for all pairs of disjoint sets $A$ and $B$, then $f+g$ has parameter $\gamma_{A,B} \geq \gamma^g_{A,B}$.
	\end{enumerate}

\end{restatable}

By combining the above  proposition with Examples~\ref{ex:diversity},~\ref{ex:|S|f(S)}, and~\ref{ex:proportionally-submod} we get the following.

\begin{restatable}{example}{divsubmod}\label{ex:nonmonotone}
	Let $f$ be a non-negative monotone submodular function. Then:
	\begin{itemize}
		\item The sum $f+g$ where $g$ is a metric diversity function satisfies $\gamma_{A,B} \geq \frac{|A|}{|A|+|B|-1}$.

		\item The sum $f(S) + |S|\cdot g(S)$ where $g$ is non-monotone submodular satisfies $\gamma_{A,B} \geq \frac{|A|+1}{|A|+|B|}$.

		\item The sum $f+g$ where $g$ is non-monotone proportionally submodular satisfies $\gamma_{A,B} \geq \frac{3 a (1 + a)}{3 a^2 + 3 a b + b^2 - 1}$, where $a=|A|$ and $b=|B|$.
	\end{itemize}

\end{restatable}

We can also combine the above result with Example~\ref{ex:f(S)g(S)} to get that the product function $(f+g) \cdot h$ satisfies $\gamma_{A,B} \geq \frac{f(A)}{f(A \cup B)} \cdot \frac{|A|}{|A|+|B|-1}$, whenever $f$ and $h$ are monotone submodular and $g$ is a metric diversity function.
This generalizes the bound from Example~\ref{ex:f(S)g(S)-apps} (b).

We note that the sum of a monotone submodular function and a metric diversity function has been previously studied~\cite{borodin2017max}. We discuss this in more detail in Section~\ref{sec:applications}.

\subsection{Additional related work}\label{sec:related-work}

The notion of weak submodularity was introduced by Das and Kempe~\cite{das2011submodular}, where they showed that the standard greedy algorithm achieves a $(1-e^{-\gamma})$-approximation for the monotone maximization problem subject to a cardinality constraint.
They provided applications to the feature selection problem for linear regression and the dictionary selection problems.
Khanna~et~al.~\cite{khanna2017scalable} showed that faster (such as distributed and stochastic) versions of the greedy algorithm also retain provable theoretical guarantees for monotone weakly submodular maximization under a cardinality constraint.
They discussed applications for the sparse linear regression problem and the support selection problem. 
Elenberg~et~al.~\cite{elenberg2017streaming} considered the
above problem in the random order streaming setting, and
provided applications to nonlinear sparse regression and interpretability of black-box neural network classifiers. Connections between weak submodularity and restricted strong convexity were shown by Elenberg~et~al.~\cite{elenberg2018restricted}, and used for high-dimensional subset selection problems.
The work of Chen~et~al.~\cite{chen2018weakly} goes beyond the cardinality constraint, and considers the monotone maximization problem subject to a matroid constraint.
They provided an approximation ratio of ${(1+1/ \gamma)}^{-2}$ for this problem, and discuss applications to video summarization, splice site detection, and black-box interpretation of images.
Gatmiry and Gomez \cite{gatmiry2018non} showed that the standard deterministic greedy algorithm also enjoys provable guarantees for the above problem, though worse than the one provided by~\cite{chen2018weakly}. They provide applications to tree-structured Gaussian graphical model estimation.
The recent work of Harshaw~et~al.~\cite{harshaw2019submodular} considers the problem $\max \{f(S)-m(S):|S| \leq k\}$, where $f$ is non-negative monotone $\gamma$-weakly submodular and $m$ is a non-negative modular function.
Using the special structure of this type of objective, they circumvented the potential roadblocks of $f-m$ being negative or non-monotone, and provided a bifactor guarantee satisfying $f(S)-m(S) \geq (1-e^{-\gamma}) f(OPT) - m(OPT)$.
In addition, they showed that this approximation ratio is tight in the value oracle model.

Non-monotone submodular maximization subject to a cardinality constraint has been studied extensively.
The first constant factor approximation for this problem was given by Lee~et~al.~\cite{lee2010maximizing}.
Since then a large series of works~\cite{buchbinder2014submodular,ene2016constrained,feldman2011unified,gharan2011submodular,gupta2010constrained,vondrak2013symmetry} have improved the approximation factor to the current best $0.385$ ratio due to Buchbinder and Feldman~\cite{buchbinder2019constrained}.
Some of the latter works, however, use an approach that involves using a continuous relaxation of the objective function and then applying rounding methods to the fractional solution.
While this approach has been extremely successful for proving strong theoretical guarantees, due to the run time they usually become impractical in real-world scenarios with large amounts of data.
In our work we use a randomized greedy algorithm proposed by Buchbinder et al.~\cite{buchbinder2014submodular}, where it is shown that this algorithm produces a $1/e$-approximation (on expectation).
On the inapproximability side, Gharan and Vondrak~\cite{gharan2011submodular} show that it is impossible to achieve a $0.491$ approximation for this problem in the value oracle model.

\subsection{Our contributions}\label{sec:contributions}

One main contribution of this work is showing that an easy-to-implement and fast randomized greedy algorithm (i.e., Algorithm~\ref{alg:random-greedy}) has provable theoretical guarantees for the problem $\max \{f(S): |S| \leq k\}$  when the function $f: 2^E \to \mathbb{R}_+$ is non-monotone weakly submodular (as defined in Section~\ref{sec:definitions}).
This is encapsulated in the following result.
To the best of our knowledge, this is the first time that weakly submodular functions are considered in the non-monotone setting.

\begin{theorem}\label{thm:non-monot-global}
	There exists an efficient randomized greedy algorithm which has an approximation ratio (on expectation) of at least $\gamma  {(1-1/\gamma k)}^{k-1}$
	for the problem of maximizing a non-negative non-monotone $\gamma$-weakly submodular function subject to a cardinality constraint of size $k$.
	This approximation ratio is asymptotically  $\gamma \cdot e^{-1/ \gamma}$ as $k \to \infty$.
	For non-negative non-monotone $\gamma$-pseudo submodular functions, the approximation ratio is of at least  $\gamma \cdot e^{-\gamma}$.
\end{theorem}

We remark that when $\gamma$ approaches to $1$, our bounds recover the $1/e$ approximation factor given in~\cite{buchbinder2014submodular} for the analysis of the same algorithm over submodular functions (i.e., the case when $\gamma=1$).

A key ingredient for analyzing non-monotone objectives is to bound the term  $\E[f(S_i \cup \mathrm{OPT})]$ with respect to $f(\mathrm{OPT})$.
For submodular functions the work of~\cite{buchbinder2014submodular} (see their Lemma 2.2 and Observation 1) bounds the above term by using the diminishing returns property, i.e., $f_A (e) \geq f_B (e)$ whenever $A \subseteq B$ and $e \notin B$. However, it is not clear how one could imitate such argument in the case of non-submodular functions. In particular, it is not obvious whether from the definition of weak submodularity, one could find a parameter $\beta >0$ satisfying some approximate version $f_A (e) \geq \beta f_B (e)$ of diminishing returns.
We circumvent this issue by analyzing the quantity $\E[f(S_i \cup \mathrm{OPT})]$ directly with respect to the execution of the algorithm (see Lemma~\ref{lem:semi1}).

Another important piece of our work is to provide a more refined analysis that allows for the submodularity ratio
to change throughout the execution of the algorithm.
This is particularly useful since many classes of functions will usually satisfy this (see for instance Section~\ref{sec:examples}). 
Our most general result (Theorem~\ref{thm:non-monot-local})  assumes some local bounds 
for the submodularity ratio
throughout the algorithm, and provides approximation guarantees based on these bounds.
Its statement is somewhat less clean to express since it depends on the notation used in Algorithm~\ref{alg:random-greedy} (which we introduce later in Section~\ref{sec:randomized-greedy}), so we defer its full presentation and discussion to Section~\ref{sec:non-monotone}. We next present some of its consequences, which lead to some of our main applications.

\begin{theorem}\label{thm:local-param}
	Assume we run the randomized greedy algorithm described in Algorithm~\ref{alg:random-greedy} on a function $f:2^E \to \mathbb{R}_+$ with parameters $\gamma_{A,B} \in [0,1]$ for any pair of disjoint sets $A,B \subseteq E$. Moreover, assume there are values  $0 \leq \gamma_i \leq 1$ for $i \in \{0,1,2,\ldots,k-1\}$ so that
	$
	\sum_{e \in \mathrm{OPT}} f_{S_i} (e) \geq  \min\{ \gamma_i \cdot f_{S_i} (\mathrm{OPT}),  f_{S_i} (\mathrm{OPT})\}
	$
	holds for any possible solution $S_i$ of the algorithm after iteration $i$.
	Then the algorithm produces (on expectation):
	\vspace*{0.05cm}
	\begin{itemize}
		\item An approximation factor of at least $1- \exp (-\frac{1}{k}\sum_{i=0}^{k-1} \gamma_i)$ if $f$ is monotone.
		\vspace*{0.05cm}
		\item An approximation factor of at least $ \frac{1}{ek} \sum_{i=0}^{k-1} \gamma_i $ if $f$ is non-monotone.
	\end{itemize}
\end{theorem}

We remark that for monotone $\gamma$-weakly submodular objectives the above result retains the $(1-e^{-\gamma})$-approximation given in~\cite{das2011submodular}. This follows by noticing that for monotone functions we always have that $\min\{ \gamma_i \cdot f_{S_i} (\mathrm{OPT}),  f_{S_i} (\mathrm{OPT})\} = \gamma_i \cdot f_{S_i} (\mathrm{OPT})$ since $f_{S_i} (\mathrm{OPT}) \geq 0$ and $\gamma_i \in [0,1]$. One can then use the $\gamma$-weak submodularity of the function to set $\gamma_i = \gamma$ for all $i$.


Combining the above theorem with the results from Section~\ref{sec:examples} leads to interesting applications. We now highlight some of them, and defer a more detailed discussion to Section~\ref{sec:applications}.
The above theorem allows us to obtain provable guarantees for some of the non-monotone objectives discussed in Section~\ref{sec:examples}. 
These include, for instance, the non-monotone functions from Example~\ref{ex:nonmonotone}, which satisfy the property $\gamma_{A,B} \in [0,1]$.



Theorem~\ref{thm:local-param} also leads to interesting results for monotone objectives.
Applying it to Example~\ref{ex:proportionally-submod}
we get a $0.197$-approximation for maximizing monotone proportionally submodular functions subject to a cardinality constraint. This improves over the current best $0.168$-approximation from~\cite{borodin2014weak,borodin2015proportionally}.
Another set of applications is obtained via Example~\ref{ex:f(S)g(S)-apps}, which allows us to get several constant factor approximations for the product of set functions. For instance, for the product $f\cdot(g+h)$ where $f,g$ are monotone submodular and $h$ is a metric diversity function, our results lead to a $0.058$-approximation. For the product $f \cdot g$ where $f$ is monotone submodular and $g$ is  monotone proportionally submodular, we get a $0.046$-approximation.
We are not aware of previous work for these problems.

\section{Approximation guarantees}

In this section we present the main theoretical contribution of this work, which is to
analyze the performance of a randomized greedy algorithm on non-monotone functions (see Section~\ref{sec:non-monotone}).
We present the analysis for monotone objectives in Section~\ref{sec:monotone}.
We next describe the randomized greedy algorithm that we use in this work.

\subsection{Randomized greedy algorithm}\label{sec:randomized-greedy}
In this section, we explain the randomized greedy algorithm introduced in the work of~\cite{buchbinder2014submodular}, where they study the problem of maximizing a non-monotone submodular function subject to a cardinality constraint.
We note that this algorithm has also been used in~\cite{chen2018weakly} for the problem of maximizing a monotone weakly submodular function subject to a matroid constraint.

Given a set function $f:2^E \to \R$ over a ground set $E$, we first add a set $D$ of $2k$ dummy elements to the ground set. That is, for any set $A \subseteq E$ and $U \subseteq D$ the function satisfies $f_A(U) = 0$.
Then, for each $1 \leq i \leq k$, we take a set of $k$ elements that maximizes the sum of the marginal gains, where in case of ties we always give preference to elements from the original ground set $E$. Finally, we choose uniformly at random one of the $k$ elements, and add it to the current solution. We summarize this procedure in Algorithm~\ref{alg:random-greedy}.

\RestyleAlgo{algoruled}
\begin{algorithm}[htb]
\footnotesize
Add a set $D$ of $2k$ dummy elements to $f$.\\
Initialize: $S_0 \leftarrow \emptyset$. \\
\For{$i = 1$ to $k$} {
    Let $M_i \subseteq (E \cup D) \setminus S_{i-1}$ be a subset of size $k$ maximizing $\sum_{e \in M_i} f_{S_{i-1}}(e)$. In case of ties between dummy elements and elements from $E$, always choose the latter.\\
    Let $e_i$ be a uniformly random element from $M_i$.\\
    $S_i \leftarrow S_{i-1}+e_i$.\\
}
\Return $S_k$.
\caption{\textsf{RandomizedGreedy}$(f,k)$}\label{alg:random-greedy}
\end{algorithm}

The algorithm is quite efficient as it makes $O(nk)$ queries to the value oracle. This is the same number of queries that the standard deterministic greedy algorithm makes.
Moreover, adding $2k$ dummy elements to the original ground set guarantees the following.

\begin{observation}\label{obs:random-greedy}
	At any iteration $1\leq i \leq k$ of the \greedy algorithm the following is satisfied:
	\begin{enumerate}
		\item $|M_i|=k$.

		\item $f_{S_{i-1}} (e_i) \geq 0$, and hence  $f(S_i) \geq f(S_{i-1})$.

		\item $\sum_{e \in M_i} f_{S_{i-1}} (e) \geq  \sum_{e \in \mathrm{OPT}} f_{S_{i-1}} (e)$.
	\end{enumerate}
\end{observation}
\begin{proof}
	The first two statements are immediate from the fact that we add $2k$ dummy elements. To see the last statement, let $\bar{M}_i$ denote a set of size $k$ containing $\mathrm{OPT} \setminus S_{i-1}$ and potentially some dummy elements (so that $|\bar{M}_i|=k$). Then, by definition of $M_i$ we have
	\[
	\sum_{e \in M_i} f_{S_{i-1}} (e) \geq \sum_{e \in \bar{M}_i} f_{S_{i-1}} (e)  =  \sum_{e \in \mathrm{OPT}} f_{S_{i-1}} (e).
	\qedhere
	\]
\end{proof}

\subsection{Analysis for monotone functions}\label{sec:monotone}

In this section, we analyze the performance of the \textsf{RandomizedGreedy} algorithm on monotone functions.
We note that we keep the term depending on the initial set $S_0$ in the approximation factor. The main reason for this is that while in many settings this will just
be the empty set, in some applications one needs to start from a non-empty initial set $S_0$  to have provable
guarantees for the parameter $\gamma_i$. (See for instance our applications for the product of set functions discussed in Section~\ref{sec:applications}.)
Then we would just run the algorithm for $k - |S_0|$ iterations.

\begin{theorem}
	\label{thm:monotone-local}
	Let $f:2^E \to \mathbb{R}_+$ be a monotone set function. Assume there are values  $0 \leq \gamma_i \leq 1$ for $i \in \{0,1,2,\ldots,k-1\}$ so that
	\[
	\sum_{e \in \mathrm{OPT}} f_{S_{i}} (e) \geq  \gamma_i \cdot f_{S_{i}} (\mathrm{OPT})
	\]
	throughout the execution of the \greedy algorithm, where $S_i$ denotes the set of chosen elements after the $ith$ iteration (i.e., $|S_i|=i$). Then at any iteration $1 \leq i \leq k$ the algorithm satisfies
	\begin{align*}
	\E[f(S_i)]  & \geq \left(1- \prod_{j=0}^{i-1} {\left(1-\frac{\gamma_j}{k} \right)} \right) \cdot f(\mathrm{OPT}) +   \prod_{j=0}^{i-1}  {\left(1-\frac{\gamma_j}{k} \right)}   \cdot \E[f(S_0)]  \\
	& \geq  \left(1- \exp \Big(-\sum_{j=0}^{i-1} \frac{\gamma_j}{k}  \Big) \right) \cdot f(\mathrm{OPT}) +   \prod_{j=0}^{i-1}  {\left(1-\frac{\gamma_j}{k} \right)}   \cdot \E[f(S_0)].
	\end{align*}
\end{theorem}
\begin{proof}
	Fix $1\leq i \leq k$ and a possible realization $S_1,S_2,\ldots,S_{i-1}$ of the algorithm  of up to iteration $i-1$.
	Then (conditioned on this event) we have
	\begin{align*}
	\E[f_{S_{i-1}} (e_i)] & = \frac{1}{k} \sum_{e \in M_i} f_{S_{i-1}} (e) \geq  \frac{1}{k} \sum_{e \in OPT} f_{S_{i-1}} (e) \geq \frac{\gamma_{i-1}}{k} f_{S_{i-1}} (OPT)  \\
	&  = \frac{\gamma_{i-1}}{k} [f(S_{i-1} \cup OPT) - f(S_{i-1})] \geq \frac{\gamma_{i-1}}{k} [f(OPT) - f(S_{i-1})],
	\end{align*}
	where the first inequality follows from Observation~\ref{obs:random-greedy}, the second inequality from the theorem's assumption, and the last inequality from non-negativity of $f$.
	We then have
	\[
	\E [f(S_i)] - f(S_{i-1}) \geq \frac{\gamma_{i-1}}{k} [f(OPT) - f(S_{i-1})],
	\]
	and rearranging the terms we get
	\[
	f(OPT) - \E [f(S_i)] \leq \Big(1-\frac{\gamma_{i-1}}{k} \Big) \Big[f(OPT) - f(S_{i-1}) \Big].
	\]
	By unfixing the realization $S_1,S_2,\ldots,S_{i-1}$ and taking expectations over all such possible realizations of the algorithm we get
	\begin{align*}
	f(OPT) - \E [f(S_i)]  &\leq  \Big(1-\frac{\gamma_{i-1}}{k} \Big) \Big[f(OPT) - \E [f(S_{i-1})] \Big] \\
	& \leq \Big(1-\frac{\gamma_{i-1}}{k} \Big)\Big(1-\frac{\gamma_{i-2}}{k} \Big) \Big[f(OPT) - \E [f(S_{i-2})] \Big]  \\
	& \leq \cdots \\
	& \leq \bigg( \prod_{j=0}^{i-1} \left(1 - \frac{\gamma_{j}}{k}\right) \bigg) [f(OPT) - \E [f(S_{0})]].
	\end{align*}
	Hence,
	\begin{align*}
	\E[f(S_i)]  & \geq \left(1- \prod_{j=0}^{i-1} {\left(1-\frac{\gamma_j}{k} \right)} \right) \cdot f(\mathrm{OPT}) +   \prod_{j=0}^{i-1}  {\left(1-\frac{\gamma_j}{k} \right)}   \cdot \E[f(S_0)]  \\
	& \geq  \left(1- \exp \Big(-\sum_{j=0}^{i-1} \frac{\gamma_j}{k}  \Big) \right) \cdot f(\mathrm{OPT}) +   \prod_{j=0}^{i-1}  {\left(1-\frac{\gamma_j}{k} \right)}   \cdot \E[f(S_0)],
	\end{align*}
	where the last inequality uses that $1-x \leq e^{-x}$ for all $ x \geq 0$.
\end{proof}

The above result now proves the first part of Theorem~\ref{thm:local-param}. This follows because by monotonicity of $f$ we have $f_{S_i} (\mathrm{OPT}) \geq 0$, and hence $\min\{ \gamma_i \cdot f_{S_i} (\mathrm{OPT}),  f_{S_i} (\mathrm{OPT})\} = \gamma_i \cdot f_{S_i} (\mathrm{OPT})$.

\subsection{Analysis for non-monotone functions}\label{sec:non-monotone}
In this section we analyze the performance of the \textsf{RandomizedGreedy} algorithm on non-monotone functions.
As mentioned in Section~\ref{sec:contributions}, a key ingredient for analyzing the non-monotone case is to bound the term  $\E[f(S_i \cup \mathrm{OPT})]$ from below with respect to $f(\mathrm{OPT})$. For monotone objectives this is trivial, since by monotonicity we always have $f(S_i \cup \mathrm{OPT}) \geq f(\mathrm{OPT})$. The techniques used in~\cite{buchbinder2014submodular} for analyzing \greedy with respect to submodular functions make use of the diminishing returns property (see their Lemma 2.2 and Observation 1). However, it is not clear how to extend those techniques for non-monotone weakly submodular functions, since it is not obvious whether they satisfy some type of approximate diminishing returns property $f_A (e) \geq \beta f_B (e)$.
Our next result circumvents this issue by analyzing the quantity $\E[f(S_i \cup \mathrm{OPT})]$ directly with respect to the execution of the algorithm.

\begin{lemma}\label{lem:semi1}
	Let $f$ be a non-negative set function.
	Assume there are numbers $0 \leq \bar{\alpha}_i \leq \bar{\beta}_i \leq k$ such that
	\[
	\sum_{u \in M_i} f_{S_{i-1} \cup \mathrm{OPT}} (u) \geq \min\{  \bar{\alpha}_i \cdot f_{S_{i-1} \cup \mathrm{OPT}} (M_i), \bar{\beta}_i \cdot f_{S_{i-1} \cup \mathrm{OPT}} (M_i) \}
	\]
	is satisfied for any choice of $M_i$ and $S_{i-1}$ throughout the execution of the \greedy algorithm.
	Then at any iteration $1 \leq i \leq k$ the algorithm satisfies
	$
	\E[f(S_i \cup \mathrm{OPT})] \geq \prod_{j=1}^i (1-\bar{\beta}_j / k ) \cdot f(\mathrm{OPT}).
	$
\end{lemma}
\begin{proof}
	Fix $1\leq i \leq k$ and an event $S_1,S_2,\ldots,S_{i-1}$ of a possible path of the algorithm up to iteration $i-1$. Then (conditioned on this event) we have
	\begin{align*}
	& \E[f(S_i \cup \mathrm{OPT})] \\
	& = f(S_{i-1} \cup \mathrm{OPT}) + \E[f_{S_{i-1} \cup \mathrm{OPT}} (u_i)]
	= f(S_{i-1} \cup \mathrm{OPT}) + \frac{1}{k} \sum_{u \in M_i} f_{S_{i-1} \cup \mathrm{OPT}} (u)\\
	&\geq  f(S_{i-1} \cup \mathrm{OPT}) + \frac{1}{k} \min\{  \bar{\alpha}_i \cdot f_{S_{i-1} \cup \mathrm{OPT}} (M_i), \bar{\beta}_i \cdot f_{S_{i-1} \cup \mathrm{OPT}} (M_i) \}.
	\end{align*}
	We now consider separately the cases where the marginal gain $f_{S_{i-1} \cup OPT} (M_i)$ is either negative or non-negative. If it is non-negative, using that $0 \leq \bar{\alpha}_i \leq \bar{\beta}_i$ we get
	\begin{equation*}
	\E[f(S_i \cup \mathrm{OPT})] \geq  f(S_{i-1} \cup \mathrm{OPT}) + \frac{\bar{\alpha}_i}{k}  f_{S_{i-1} \cup OPT} (M_i) \geq f(S_{i-1} \cup \mathrm{OPT}).
	\end{equation*}
	If on the other hand, the marginal gain is negative, then
	\begin{align*}
	\E[f(S_i \cup \mathrm{OPT})] &\geq f(S_{i-1} \cup \mathrm{OPT}) + \frac{\bar{\beta}_i}{k} f_{S_{i-1} \cup \mathrm{OPT}} (M_i)\\
	&= f(S_{i-1} \cup \mathrm{OPT}) + \frac{\bar{\beta}_i}{k} [f(S_{i-1} \cup \mathrm{OPT} \cup M_i) - f(S_{i-1} \cup \mathrm{OPT})]\\
	&\geq f(S_{i-1} \cup \mathrm{OPT}) - \frac{\bar{\beta}_i}{k} f(S_{i-1} \cup \mathrm{OPT})
	= \Big[1- \frac{\bar{\beta}_i}{k} \Big] f(S_{i-1} \cup \mathrm{OPT}),
	\end{align*}
	where the last inequality follows from non-negativity. Thus, for each possible fixed realization $S_1,S_2,\ldots,S_{i-1}$ of the algorithm up to iteration $i-1$ we have
	\begin{equation*}
	\E[f(S_i \cup \mathrm{OPT})] \geq \Big[1- \frac{\bar{\beta}_i}{k} \Big] f(S_{i-1} \cup \mathrm{OPT}).
	\end{equation*}

	By unconditioning on the event $S_1,S_2,\ldots,S_{i-1}$, and taking the expectation over all such possible events we get:
	\begin{align*}
	\E[f(S_i \cup \mathrm{OPT})] & \geq \Big[1- \frac{\bar{\beta}_i}{k} \Big] \E[f(S_{i-1} \cup \mathrm{OPT})]
	\geq \Big[1-\frac{\bar{\beta}_i}{k} \Big] \Big[1- \frac{\bar{\beta}_{i-1}}{k} \Big]   \E[f(S_{i-2} \cup \mathrm{OPT})] \\
	&\geq \cdots \geq
	\prod_{j=1}^i \Big[1- \frac{\bar{\beta}_j}{k} \Big] \E[f(S_{0} \cup \mathrm{OPT})] = \prod_{j=1}^i \Big[1- \frac{\bar{\beta}_j}{k}  \Big] f(\mathrm{OPT}).
	\qedhere
	\end{align*}
\end{proof}

For submodular functions the above result becomes $\E[f(S_i \cup \mathrm{OPT})] \geq {(1-1/k)}^i \cdot f(\mathrm{OPT})$, since  we can take $\bar{\alpha}_i = \bar{\beta}_i =1$ for all $i$. We remark that this
matches the bound provided in~\cite{buchbinder2014submodular} for submodular functions (see their Observation 1). 
We now prove our main result.

\begin{theorem}\label{thm:non-monot-local}
	Let $f:2^E \to \mathbb{R}_+$ be a set function.
	Assume there are values $0 \leq \bar{\alpha}_i \leq \bar{\beta}_i \leq k$ and $0 \leq {\alpha}_i \leq {\beta}_i \leq k$ such that
	\[
	\sum_{u \in M_i} f_{S_{i-1} \cup \mathrm{OPT}} (u) \geq  \min\{  \bar{\alpha}_i \cdot f_{S_{i-1} \cup \mathrm{OPT}} (M_i), \bar{\beta}_i \cdot f_{S_{i-1} \cup \mathrm{OPT}} (M_i) \}
	\]
	and
	\[
	\sum_{e \in OPT} f_{S_{i-1}} (e) \geq \min\{  \alpha_{i-1} \cdot f_{S_{i-1}} (\mathrm{OPT}), \beta_{i-1} \cdot f_{S_{i-1}} (\mathrm{OPT}) \}
	\]
	is satisfied for any choice of $M_i$ and $S_{i-1}$ throughout the execution of the \greedy algorithm.
	Then at any iteration $1 \leq i \leq k$ the algorithm satisfies
	\[
	\E[f(S_i)] \geq  \left( \prod_{j=1}^{i-1} \min \Big\{1-\frac{\bar{\beta}_j}{k}, 1-\frac{\alpha_j}{k} \Big\}  \right) \cdot \Big( 	 \sum_{j=0}^{i-1} \frac{\alpha_j}{k}    \Big) \cdot  f(\mathrm{OPT}).
	\]
\end{theorem}
\begin{proof}
	Fix $1\leq i \leq k$ and an event $S_1,S_2,\ldots,S_{i-1}$ of a possible realization of the algorithm up to iteration $i-1$. Then (conditioned on this event) we have
	\begin{align*}
	\E[f_{S_{i-1}} (e_i)] & = \frac{1}{k} \sum_{e \in M_i} f_{S_{i-1}} (e) \geq  \frac{1}{k} \sum_{e \in OPT} f_{S_{i-1}} (e) \\
	&\geq \frac{1}{k} \min\{  \alpha_{i-1} \cdot f_{S_{i-1}} (\mathrm{OPT}), \beta_{i-1} \cdot f_{S_{i-1}} (\mathrm{OPT}) \},
	\end{align*}
	where the first inequality follows from Observation~\ref{obs:random-greedy}, and the second inequality from the theorem's assumption.
	We now consider separately the cases where the marginal gain $f_{S_{i-1}} (\mathrm{OPT})$ is either negative or non-negative. If it is non-negative,  using that $0 \leq {\alpha}_i \leq {\beta}_i$ we get 
	\[
	\E[f_{S_{i-1}} (e_i)] \geq \frac{\alpha_{i-1}}{k}  \cdot f_{S_{i-1}} (\mathrm{OPT}).
	\]
	If it is negative we get
	\[
	\E[f_{S_{i-1}} (e_i)] \geq 0 \geq \frac{\alpha_{i-1}}{k} \cdot f_{S_{i-1}} (\mathrm{OPT}),
	\]
	where the first inequality follows from  Observation~\ref{obs:random-greedy}.  It then follows that for any fixed possible realization $S_1,S_2,\ldots,S_{i-1}$ of the algorithm of up to iteration $i-1$ we have
	\begin{equation}
	\label{eq:greedy-bound}
	\E[f_{S_{i-1}} (e_i)] \geq \frac{\alpha_{i-1}}{k} \cdot f_{S_{i-1}} (\mathrm{OPT}) = \frac{\alpha_{i-1}}{k} [f(S_{i-1} \cup \mathrm{OPT}) - f(S_{i-1})].
	\end{equation}

	We now unfix the realization $S_1,S_2,\ldots,S_{i-1}$ and take expectations over all such possible realizations of the algorithm.
	\begin{align}
	\label{eq:thm}
	\E[f(S_i)] &= \E[f(S_{i-1}) + f_{S_{i-1}} (e_i)] = \E[f(S_{i-1})] + \E[f_{S_{i-1}} (e_i)]  \nonumber \\ &
	\geq \E[f(S_{i-1})] + \frac{\alpha_{i-1}}{k} \E[f(S_{i-1} \cup \mathrm{OPT}) - f(S_{i-1})]   \nonumber \\
	& = \Big[1-\frac{\alpha_{i-1}}{k}\Big]  \E[f(S_{i-1})] + \frac{\alpha_{i-1}}{k} \E[f(S_{i-1} \cup \mathrm{OPT})]  \nonumber  \\ &
	\geq \Big[1-\frac{\alpha_{i-1}}{k}\Big]  \E[f(S_{i-1})] + \frac{\alpha_{i-1}}{k} \prod_{j=1}^{i-1} \Big[1-\frac{\bar{\beta}_j}{k}\Big]  f(\mathrm{OPT}),
	\end{align}
	where the first inequality follows from Equation~\eqref{eq:greedy-bound} and the last inequality follows from Lemma~\ref{lem:semi1} (which we can use due to the theorem's assumptions).

	We are now ready to prove the statement of the theorem using induction on the value of $1 \leq i \leq k$.
	The base case $i=1$ claims that $\E[f(S_1)] \geq (\alpha_0 / k) \cdot f(\mathrm{OPT})$. This follows from Equation \eqref{eq:thm} by setting $i=1$ and using that $f(S_0) = f(\emptyset) \geq 0$.

	Now let $1<i \leq k$ be arbitrary, and assume that the claim is true for all values $1 \leq i'<i$; we show it is also true for $i$. Using Equation~\eqref{eq:thm} and the induction hypothesis we get
	\begin{align*}
	& \E[f(S_i)] \geq \Big[1-\frac{\alpha_{i-1}}{k}\Big]  \E[f(S_{i-1})] + \frac{\alpha_{i-1}}{k} \prod_{j=1}^{i-1} \Big[1-\frac{\bar{\beta}_j}{k}\Big]  f(\mathrm{OPT})\\
	&\geq  \Bigg[  \Big[1-\frac{\alpha_{i-1}}{k}\Big]  \left( \prod_{j=1}^{i-2} \min \Big\{1-\frac{\bar{\beta}_j}{k}, 1-\frac{\alpha_j}{k} \Big\}  \right)    \cdot \Big( \sum_{j=0}^{i-2} \frac{\alpha_j}{k}  \Big)     + \frac{\alpha_{i-1}}{k} \prod_{j=1}^{i-1} \Big[1- \frac{\bar{\beta}_j}{k} \Big] \Bigg] f(\mathrm{OPT}) \\
	&  \geq \left( \prod_{j=1}^{i-1} \min \Big\{1-\frac{\bar{\beta}_j}{k}, 1-\frac{\alpha_j}{k} \Big\}  \right)  \cdot \Big( \big( \sum_{j=0}^{i-2} \frac{\alpha_j}{k}  \big) + \frac{\alpha_{i-1}}{k}  \Big) \cdot  f(\mathrm{OPT}) \\
	&=  \left( \prod_{j=1}^{i-1} \min \Big\{1-\frac{\bar{\beta}_j}{k}, 1-\frac{\alpha_j}{k} \Big\}  \right) \cdot \Big( 	 \sum_{j=0}^{i-1} \frac{\alpha_j}{k}    \Big) \cdot  f(\mathrm{OPT}).
	\qedhere
	\end{align*}
\end{proof}

\subsubsection{Proof of Theorems~\ref{thm:non-monot-global} and \ref{thm:local-param}}
The above result leads to several interesting consequences by choosing the values of the parameters $\alpha_i, \bar{\alpha}_i, \beta_i , \bar{\beta}_i$ appropriately. For instance, for non-monotone $\gamma$-weakly submodular functions we have $\alpha_i = \bar{\alpha}_i = \gamma$ and $\beta_i = \bar{\beta}_i = 1/ \gamma$ for all $i$. Hence we immediately get an approximation of $\gamma  {(1-1/\gamma k)}^{k-1}$, which is asymptotically  $\gamma e^{-1/ \gamma}$ as $k \to \infty$. In a similar fashion, for $\gamma$-pseudo submodular functions we can take $\alpha_i = \bar{\alpha}_i = \beta_i = \bar{\beta}_i = \gamma$ for all $i$, leading to an approximation factor of $\gamma  {(1-\gamma / k)}^{k-1} \geq \gamma e^{-\gamma}$. This now proves Theorem~\ref{thm:non-monot-global}.

One can now also prove the second part of  Theorem~\ref{thm:local-param} as follows. First, if the function has a parameter that satisfies $0 \leq \gamma_{A,B} \leq 1$ (such as in Example~\ref{ex:nonmonotone}), we immediately get that $\alpha_i, \bar{\alpha}_i, \beta_i , \bar{\beta}_i \leq \max_{A \cap B = \emptyset} \gamma_{A,B} \leq 1$. 
Hence $\prod_{j=1}^{k-1} \min \{1-\alpha_j/k, 1-\bar{\beta}_j / k \}  \geq {[1-1/k]}^{k-1} \geq 1/e$.
In addition, using the assumptions from Theorem~\ref{thm:local-param} one can take $\alpha_i = \gamma_i$ and $\beta_i = 1$ in Theorem~\ref{thm:non-monot-local};
leading to an approximation factor of $(1/ek) \cdot  \sum_{j=0}^{k-1} \gamma_j$ as desired.

Theorem~\ref{thm:non-monot-local} becomes particularly useful to prove tighter guarantees for some of the examples discussed in Section~\ref{sec:examples}, which have a parameter $\gamma_{A,B}$ that changes throughout the algorithm. We discuss this and applications for monotone objectives in the next section.

\section{Applications}\label{sec:applications}

We now present some applications for our results. We discuss the monotone case first.

For monotone functions it is clear that the \greedy algorithm always selects $k$ elements from the original ground set $E$ (i.e., it never chooses dummy elements). In particular, the current solution $S_i$ at iteration $i$ always has $i$ elements from $E$, while  $OPT \setminus S_i$ is a set containing at most $k$ elements. We can use this, together with the results from Section~\ref{sec:examples}, to compute a lower bound for a parameter $\gamma_i \geq 0$ that satisfies
$
\sum_{e \in \mathrm{OPT}} f_{S_i} (e) \geq  \gamma_i \cdot f_{S_i} (\mathrm{OPT}).
$
For instance, one can take
\vspace*{-0.1cm}
\[
\gamma_i = \min_{|A|=i, \ 1\leq |B| \leq k, \  A \cap B = \emptyset} \gamma_{A,B}.
\]
We then immediately get a provable approximation ratio of at least $1- \exp (-\frac{1}{k}\sum_{i=0}^{k-1} \gamma_i)$ via Theorem~\ref{thm:local-param} (or Theorem~\ref{thm:monotone-local}). 

For monotone proportionally submodular functions, Example~\ref{ex:proportionally-submod} gives a bound of $\gamma_{A,B} \geq \frac{3 a (1 + a)}{3 a^2 + 3 a b + b^2 - 1}$ where $a=|A|$ and $b=|B|$. Hence $\gamma_i \geq \frac{3 i (1 + i)}{3 i^2 + 3 i k + k^2 - 1}$ for $i\in \{0,1,\ldots,k-1\}$.
By plugging this into Theorem~\ref{thm:local-param} we get an expression that does not seem to have a closed form, but that numerically converges from above to $0.197$.
This improves over the approximation factor of $0.168$ given in~\cite{borodin2015proportionally} for the same problem (they give it as a $5.95$-approximation since they express approximation factors as numbers greater than $1$).

\begin{theorem}
\label{thm:prop-submod}
	There is an efficient $0.197$-approximation for the problem of maximizing a non-negative monotone proportionally submodular function subject to a cardinality constraint.
\end{theorem}

Our next application is for the product of monotone set functions. First, let us consider the case $f \cdot g$ where $f$ is submodular and $g$ is either submodular, metric diversity, or proportionally submodular. Example~\ref{ex:f(S)g(S)-apps} provides explicit bounds for the parameter $\gamma_{A,B}$ of these product functions. We have $ \gamma_{A,B} \geq ( f(A) / f(A \cup B) ) \cdot \gamma^g_{A,B}$ where the latter term denotes the parameter of the function $g$. Hence, we need to lower bound the term $f(A) / f(A \cup B)$. We can do this as follows. 
One can show that for submodular functions, if there is a set $S_f$ satisfying
$f(S_f) \geq \alpha \cdot \max_{|S| \leq k} f(S)$ then $ f(A) / f(A \cup B) \geq \alpha / (1+\alpha)$ for any set $A \supseteq S_f$ and any set $B$ of size at most $k$ (see Claim~\ref{claim:submodular-product} in the Appendix).
We can then take $S_0 = S_f$ as the initial set and run the \greedy algorithm during $k - |S_0|$ iterations (to get a set of size $k$), with a guarantee that the parameter of the product function satisfies $\gamma_{A,B} \geq \alpha \cdot \gamma^g_{A,B}$. This leads to approximation guarantees of $1- \exp (-\frac{1}{k}\sum_{i=k/2}^{k-1} \alpha \cdot \gamma^g_i)$, where $\gamma^g_i$ denotes the parameter $\gamma_i$ of the function $g$.

For submodular functions, we can run the standard greedy algorithm on $f$ during $k/2$ iterations to find a set $S_f \subseteq E$ of size $k/2$ satisfying $f(S_f) \geq (1-e^{-1/2}) \cdot \max_{|S| \leq k} f(S)$. Combining this with the fact that submodular functions have $\gamma_i \geq 1$, the sum of submodular and metric diversity has $\gamma_i \geq \frac{i}{i+k-1}$, and proportionally submodular functions have $\gamma_i \geq \frac{3 i (1 + i)}{3 i^2 + 3 i k + k^2 - 1}$ for $i\in \{0,1,\ldots,k-1\}$, one can obtain the following approximation guarantees.

\begin{theorem}\label{thm:app-product}
	Let $f,g$ and $h$ be non-negative and monotone. If $f$ is submodular, then:
	\begin{itemize}
		\item There is an approximation (on expectation) of $0.131$ for $f \cdot g$ when $g$ is submodular.

		\item There is an approximation (on expectation) of $0.058$ for $f \cdot (g + h)$  when $g$ is a metric diversity function and $h$ is submodular.

		 \item There is an approximation (on expectation) of $0.046$ for $f \cdot g$ when $g$ is proportionally submodular.
	\end{itemize}
\end{theorem}


We are not aware of previous work for the product of set functions that we can compare our results to. However, when the functions are monotone, a natural baseline can be obtained by taking the set $S := S_f \cup S_g$ where $S_f$ is obtained by running the greedy algorithm for $\max_{|S| \leq k/2} f(S)$, and similarly $S_g$ is obtained by running the greedy algorithm for $\max_{|S| \leq k/2} g(S)$. Then if $f(S_f) \geq \alpha_f \cdot \max_{|S| \leq k} f(S)$ and $g(S_g) \geq \alpha_g \cdot \max_{|S| \leq k} g(S)$, we get that $(f\cdot g) (S_f \cup S_g) \geq \alpha_f \cdot \alpha_g \cdot fg(OPT)$. In the case of the above functions we get the following guarantees for $\alpha$ after running the greedy algorithm for $k/2$ iterations:  for a submodular function we get $\alpha \geq 1-e^{-1/2} $ via the standard greedy algorithm analysis, for the sum of submodular and metric diversity we get $\alpha \geq 1/8$ via the analysis from~\cite{borodin2017max}, and for proportionally submodular we get $\alpha \geq 0.05 $ via the analysis using Example~\ref{ex:proportionally-submod} and Theorem~\ref{thm:local-param} (which improves over the previous analysis given in~\cite{borodin2015proportionally}).
This leads to the following baselines (though there is room for optimizing the sizes of $S_f$ and $S_g$): a 0.155-approximation for the product of two submodular functions, a $0.049$-approximation for the product of a submodular function and the sum of submodular and metric diversity, and a $0.019$-approximation for the product of a submodular function and a proportionally submodular function. 

We note that our choice of cardinality $k/2$ for the initial set $S_0$ of the algorithm, and for the sets $S_f$ and $S_g$ used in the baselines, may not be optimal.
For the sake of consistency and to keep the argument as clean as possible, we used the same cardinality for all of them.

By using a similar argument to the one from Theorem~\ref{thm:app-product} one can also get constant factor approximations in the case where $f$ is a metric diversity function. This follows since if $S_f \subseteq E$ satisfies $f(S_f) \geq \alpha \cdot \max_{|S| \leq k} f(S)$, then $f(A) / f(A \cup B) \geq \alpha / (5 + \alpha) $ for any set $A \supseteq S_f$ and any set $B$ of size at most $k$ (see Claim~\ref{claim:diversity-product} in the Appendix).
The fact that this bound is worse than for submodular functions is expected, since $f$ is supermodular and hence $f_A(e) \leq f_B(e)$ whenever $A \subseteq B \subseteq E$ and $e \notin B$.

We now discuss the non-monotone case. While for monotone functions the algorithm always chooses $k$ elements from the original ground set $E$ (i.e., it never picks dummy elements), this may not be the case for non-monotone objectives. That is, for non-monotone objectives we have $S_k \subseteq E \cup D$. Hence, we cannot just directly plug the bounds for $\gamma_{A,B}$ from Section~\ref{sec:examples}, since these depend on the number of elements from $E$ that the current solution $S_i$ has. Our next result gives a guarantee with respect to the number of elements from $E$ that the algorithm picks.


\begin{proposition}
\label{prop:app-nonmonotone}
	Let $f:2^E \to \mathbb{R}_+$ be a set function with parameters $\gamma_{A,B} \in  [0,1]$ satisfying that $\gamma_{A,B} \geq \gamma_{A',B}$ whenever $|A| \geq |A'|$. Then, if the \greedy algorithm picks $m$ elements from the original ground set $E$ (i.e., not dummy elements), its output $S_k \subseteq E \cup D$ satisfies
	\begin{equation*}
	\E[f(S_k)] \geq \frac{1}{k e}  \Big[(k-m) \bar{\gamma}_1 +\sum_{i=0}^{m-1} \bar{\gamma}_{i}  \Big] \cdot f(\mathrm{OPT}),
	\end{equation*}
	where $\bar{\gamma}_i = \min\{ \gamma_{A,B} :|A|=i, \ 1 \leq |B| \leq k, \  A \cap B = \emptyset \}$.
\end{proposition}
\begin{proof}
    First note that since $\gamma_{A,B} \in [0,1]$ we also have that $\bar{\gamma}_i \in [0,1]$.
    We show that after iteration $i$, any realization of the algorithm (conditioned on the event that $m$ elements from $E$ are selected) satisfies 
    $
	\sum_{e \in \mathrm{OPT}} f_{S_i} (e) \geq  \min\{ \gamma_i \cdot f_{S_i} (\mathrm{OPT}),  f_{S_i} (\mathrm{OPT})\}
	$ 
	where 
	$$
    \gamma_i=
    \begin{cases}
    \bar{\gamma}_0  & \mbox{if } i=0, \\
    \bar{\gamma}_1 & \mbox{if } 1 \leq i \leq k-m, \\
    \bar{\gamma}_{i-k+m} & \mbox{if }  k-m+1 \leq i \leq k-1.
    \end{cases}
    $$
	Then the desired result follows from Theorem~\ref{thm:local-param}, since
    $
	\sum_{i=0}^{k-1} \gamma_i = \bar{\gamma}_0 + (k-m) \bar{\gamma}_1 +\sum_{i=1}^{m-1} \bar{\gamma}_{i} .
	$
    
    First note that when $f_{S_i} (\mathrm{OPT}) < 0$ we have
    \[
    \sum_{e \in \mathrm{OPT}} f_{S_i} (e) \geq \gamma_{ {\scriptscriptstyle S_i \cap E, \mathrm{OPT} \setminus S_i}} \cdot f_{S_i} (\mathrm{OPT}) \geq  f_{S_i} (\mathrm{OPT}) = \min\{ \gamma_i \cdot f_{S_i} (\mathrm{OPT}),  f_{S_i} (\mathrm{OPT})\},
    \]
    where the first inequality follows from the definition of the parameter $\gamma_{A,B}$, the second inequality follows since $\gamma_{ {\scriptscriptstyle S_i \cap E, \mathrm{OPT} \setminus S_i}} \in [0,1]$ and $f_{S_i} (\mathrm{OPT}) < 0$, and similarly the last equality follows since $\gamma_i \in [0,1]$ and $f_{S_i} (\mathrm{OPT}) < 0$. 
    Hence the inequality
    $
	\sum_{e \in \mathrm{OPT}} f_{S_i} (e) \geq  \min\{ \gamma_i \cdot f_{S_i} (\mathrm{OPT}),  f_{S_i} (\mathrm{OPT})\}
	$ 
	is always satisfied when $f_{S_i} (\mathrm{OPT}) < 0$. In the case where $f_{S_i} (\mathrm{OPT}) \geq 0$ we get
    \begin{equation}
    \label{eq:prop-nonmonotone}
        \sum_{e \in \mathrm{OPT}} f_{S_i} (e) \geq \gamma_{ {\scriptscriptstyle S_i \cap E, \mathrm{OPT} \setminus S_i}} \cdot f_{S_i} (\mathrm{OPT}) \geq \bar{\gamma}_{ {\scriptscriptstyle |S_i \cap E|} } \cdot f_{S_i} (\mathrm{OPT}),
    \end{equation}
    where the last inequality follows from the definition of $\bar{\gamma}_i$ and the fact that $|OPT \setminus S_i| \leq k$. We lower bound the term $|S_i \cap E|$ as follows. 
    
	We can always assume that in the first iteration the algorithm picks an element from the original ground set. This is because $f(\emptyset)= 0$ and $f(e) \geq 0$ for all $e \in E$ by non-negativity of $f$. Hence there is always a choice of $k$ elements from the original ground set for the candidate set $M_1$.
	Now, since $\gamma_{A,B} \geq \gamma_{A',B}$ whenever $|A| \geq |A'|$, we have that the values $\bar{\gamma}_i$ are non-decreasing. It then follows that the worst scenario occurs when the algorithm picks the remaining $m-1$ non-dummy elements in the last $m-1$ iterations. In this case, we get that $|S_0 \cap E|=0$, $|S_i \cap E| = 1$ for $1 \leq i \leq k-m$, and $|S_i \cap E| = i-k+m$ for $k-m+1 \leq i \leq k$. Combining this with Equation \eqref{eq:prop-nonmonotone} leads to the desired result.
\end{proof}

The above result can be used to obtain bounds for some of the examples discussed in Section~\ref{sec:examples} that satisfy $0 \leq \gamma_{A,B} \leq 1$ and have non-decreasing values $\gamma_{A,B}$ as a function of $|A|$, such as those from Example~\ref{ex:nonmonotone}. We discuss this next, where we state the approximation guarantees for the case the algorithm selects at least $k/2$ elements from $E$.

\begin{corollary}
	Let $f:2^E \to \mathbb{R}_+$ be a (non-monotone) set function, and assume the \greedy algorithm picks at least $k/2$ elements from $E$. Then its output $S_k \subseteq E \cup D$ satisfies the following guarantees:
	\vspace{0.1cm}
	\begin{enumerate}[(a)]
		\item If $f=g+h$ where $g$ is monotone submodular and $h$ is non-monotone proportionally submodular, then $\E[f(S_k)] \geq  \frac{0.05}{e} \cdot f(\mathrm{OPT})$.

		\vspace{0.1cm}

		\item If  $f(S):=g(S) + |S| \cdot h(S)$ where $g$ is monotone submodular and $h$ is non-monotone submodular, then $\E[f(S_k)] \geq  \frac{0.09}{e} \cdot f(\mathrm{OPT})$.
	\end{enumerate}
\end{corollary}
\begin{proof}
From Example~\ref{ex:nonmonotone} we know that the function $f$ from part (b) satisfies $\bar{\gamma}_i \geq (i+1) / (i+k)$, while the function from part (a) satisfies $\bar{\gamma}_i \geq (3i(i+1)) / (3i^2+3ik+k^2-1)$. Plugging these values into Proposition~\ref{prop:app-nonmonotone} leads to the desired bounds.
\end{proof}

\section{Conclusion}

In this paper we introduced a natural generalization of weak submodularity for non-monotone functions. We showed that a randomized greedy algorithm has provable approximation guarantees for maximizing these functions subject to a cardinality constraint. 
We also provided a fine-grained analysis that allows  the submodularity ratio to change throughout the algorithm. We discussed applications of our results for monotone and non-monotone functions.

It is open whether the $(\gamma \cdot e^{-1/\gamma})$-approximation is asymptotically tight for the maximization problem subject to a cardinality constraint. Another natural direction for future work is to consider the non-monotone maximization problem under more general constraints, such as matroids or knapsacks.

\bibliography{references}

\appendix

\section{Proofs for Section~\ref{sec:examples}}\label{sec:appendix-examples-proof}

In this section we present the proofs for the results discussed in Section~\ref{sec:examples}. Given that the argument for proportionally submodular functions (i.e., Example~\ref{ex:proportionally-submod}) is much longer and involved than the rest, we discuss it in a different subsection.

\subsection{Proof of Example~\ref{ex:proportionally-submod} from Section~\ref{sec:examples}}\label{sec:proportionally-submod}

Recall that a function $f: 2^E \to \mathbb{R}_+$ is \emph{proportionally submodular}~\cite{borodin2014weak,borodin2015proportionally} if
\begin{align}
	|S|f(T) + |T|f(S) \geq |S \cap T|f(S \cup T) + |S \cup T|f(S \cap T)
	\label{eq:proportional-submodularity}
\end{align}
for every $S,T \subseteq E$.
The next three results show the desired bound for $\gamma_{A,B}$.

\begin{lemma}\label{lem:base-decrement}
	Let $f: 2^E \to \mathbb{R}_+$ be a proportionally submodular function. Then for any set $A \subseteq E$ of size $a$ and elements $e,e' \in E \setminus A$, we have
	\[
	f_{A \cup \{e'\}}(e) \leq \frac{1}{a}f_A(e') + \left(1+\frac{1}{a}\right) f_A(e).
	\]
\end{lemma}
\begin{proof}
	Substituting $S = A \cup \{e\}$ and $T = A \cup \{e'\}$ in~\eqref{eq:proportional-submodularity}, we have
	\[
	(a+1) f(A \cup \{e'\}) + (a+1) f(A \cup \{e\}) \geq a f(A \cup \{e,e'\}) + (a+2) f(A),
	\]
	which implies
	\[
	f_A(e') + (a+1) f_A(e) \geq a f_{A \cup \{e'\}}(e).
	\qedhere
	\]
\end{proof}

\begin{lemma}\label{lem:base-shrink}
	Let $f: 2^E \to \mathbb{R}_+$ be a proportionally submodular function. Then for any set $A$ of size $a$ and elements $e,e_1,\ldots,e_b \in E \setminus A$, we have
	\[
	f_{A \cup \{e_1,\ldots,e_b\}}(e) \leq \frac{a+b}{a}f_A(e) + \frac{a+b}{a(a+1)} \sum_{i=1}^b f_A(e_i).
	\]
\end{lemma}
\begin{proof}
	Let $d := a + b$.
	We successively apply Lemma~\ref{lem:base-decrement} to $f_{A \cup \{e_1,\ldots,e_b\}}(e)$.
	Each application of the lemma generates two terms so the whole applications can be seen as a binary tree.
	For the nodes in the $i$-th level, we apply Lemma~\ref{lem:base-decrement} to remove $e_i$ from the base.
	We get
	\[
	f_{A \cup \{e_1,\ldots,e_b\}}(e) \leq c \cdot f_A(e) + \sum_{i=1}^b c_i \cdot f_A(e_i),
	\]
	where
	\begin{align*}
		c & = \prod_{i=1}^b \left(1+\frac{1}{d-i}\right), \\
		c_1 & = \frac{1}{d-1} \prod_{i=2}^b \left(1+\frac{1}{d-i}\right)
		= \left(1 + \frac{2}{d-1} - 1 - \frac{1}{d-1}\right) \prod_{i=2}^b \left(1+\frac{1}{d-i}\right), \\
		c_2 & = \left(1+\frac{2}{d-1}\right)\frac{1}{d-2}\prod_{i=3}^b \left(1+\frac{1}{d-i}\right)
		= \left(1+\frac{2}{d-1}\right)\left(1+\frac{2}{d-2} - 1-\frac{1}{d-2} \right)\prod_{i=3}^b \left(1+\frac{1}{d-i}\right), \\
		& \vdots \\
		c_b & = \prod_{i=1}^{b-1} \left(1+\frac{2}{d-i}\right) \frac{1}{d-b}
		= \prod_{i=1}^{b-1} \left(1+\frac{2}{d-i}\right) \left(1+\frac{2}{d-b}-1-\frac{1}{d-b} \right).
	\end{align*}
	In addition, we have
	\[
	\sum_{i=1}^b c_i = \prod_{i=1}^{b} \left(1+\frac{2}{d-i}\right) - \prod_{i=1}^{b} \left(1+\frac{1}{d-i}\right) = \frac{d(d+1)}{(d-b)(d-b+1)} - \frac{d}{d-b} = \frac{db}{(d-b)(d-b+1)}.
	\]
	Note that the ordering of the elements $e_1,\ldots,e_b$ is arbitrary.
	Hence, summing up the $b$ inequalities obtained by rotating the ordering and dividing it by $b$, we get
	\[
	f_{A \cup \{e_1,\ldots,e_b\}}(e)
	\leq c \cdot f_A(e) + \sum_{i=1}^b \left(\frac{1}{b}\sum_{j=1}^b c_j\right) f_A(e_i)
	= \frac{d}{d-b} f_A(e) + \frac{d}{(d-b)(d-b+1)} \sum_{i=1}^b f_A(e_i).
	\]
	Substituting $d = a + b$ leads to the desired result.
\end{proof}

\begin{theorem}\label{thm:proportionally-submod}
	A non-negative proportionally submodular function $f: 2^E \to \mathbb{R}_+$ has
	\[
	\gamma_{A,B} \geq \frac{3 a (1 + a)}{3 a^2 + 3 a b + b^2 - 1}
	\]
	for any two disjoint sets $A,B \subseteq E$, 
	where $a=|A|$ and $b=|B|$.
\end{theorem}
\begin{proof}
    For $|B|=1$ we get $\gamma_{A,B}=1$, which clearly holds since in this case we trivially have $\sum_{e\in B}f_A(e)=f_A(B)$.
	Now let $B =\{ e_1,\ldots,e_b \}$ with $b \geq 2$.
	By Lemma~\ref{lem:base-shrink}, we have
	\begin{align*}
		f_A (B) & = f_A(\{e_1,\ldots,e_b\})
		= \sum_{i=1}^b f_{A \cup \{e_1,\ldots,e_{i-1}\}}(e_i) \\ &
		\leq \sum_{i=1}^b  \left( \frac{a+i-1}{a}f_A(e_i) + \frac{(a+i-1)}{a(a+1)}  \sum_{j=1}^{i-1} f_A(e_j) \right)
		\leq \sum_{i=1}^b c_i f_A(e_i),
	\end{align*}
	where
	\begin{align*}
		c_i
		= \frac{a+i-1}{a} + \sum_{j=i+1}^{b}  \frac{a+j-1}{a(a+1)}.
	\end{align*}
	Note that
	\[
	\sum_{i=1}^b c_i = \frac{-b + 3 a^2 b + 3 a b^2 + b^3}{3 a (1 + a)} =: b \cdot T_{a,b}.
	\]
	As the ordering of the elements $e_1,\ldots,e_b$ is arbitrary, we have
	\[
	f_A(B)
	\leq \sum_{i=1}^b \left(\frac{1}{b} \sum_{j=1}^b c_j\right) f_A(e_i)
	=
	\sum_{i=1}^b T_{a,b }f_A(e_i) = T_{a,b } \sum_{i=1}^b f_A(e_i).
	\]
	Taking $\gamma_{A,B} := 1/T_{a,b}$ leads to the desired result.
\end{proof}

\subsection{Proofs for other examples in Section~\ref{sec:examples}}\label{sec:examples-proofs}

\metrics*
\begin{proof}
	Let $A,B \subseteq E$ be two disjoint sets with $|A|=a$ and $|B|=b$.
	Let $d(A,B)$ denote the sum of distances between elements in $A$ and $B$. That is, $d(A,B) = \sum_{u \in A, v \in B} d(u,v)$. Then one can write $f(A \cup B) = f(A) + f(B) + d(A,B)$. Hence  $f_A(B) = f(A \cup B) - f(A) = f(B) + d(A,B)$ and $\sum_{e \in B} f_A(e) = \sum_{e \in B} d(A,e) = d(A,B)$.

	We can now use a lemma from~\cite{ravi1994heuristic} that shows that $a \cdot f(B) \leq (b-1) \cdot d(A,B)$ whenever $A$ and $B$ are disjoint. This then implies
	\[
	f_A (B) = f(B) + d(A,B) \leq \left( \frac{b-1}{a} + 1 \right) d(A,B) = \left( \frac{a+b-1}{a} \right) \sum_{e \in B} f_A(e).
	\]
\end{proof}

\product*
\begin{proof}
	Let $A,B \subseteq E$ be two disjoint sets. Then
	\begin{align*}
	h_A(e) & = f(A+e)g(A+e)- f(A) g(A) \\ & =
	f(A+e)g(A+e) - f(A+e)g(A) + f(A+e)g(A) - f(A) g(A)  \\
	& = f(A+e)g_A(e) + f_A(e)g(A).
	\end{align*}
	Hence
	\begin{align}
	\label{eq1}
	\sum_{e \in B} h_A(e)
	& = \sum_{e \in B} \Bigl(f(A+e)g_A(e) + f_A(e)g(A)\Bigr)
	\geq
	\sum_{e \in B} \Bigl(f(A)g_A(e) + f_A(e)g(A)\Bigr) \nonumber \\
	& \geq f(A) \cdot \gamma^g_{A,B} \cdot g_A(B) + \gamma^f_{A,B} \cdot f_A(B) \cdot g(A),
	\end{align}
	where the first inequality follows from the monotonicity of $f$ and $g$ (since then $f(A+e) \geq f(A)$ and $g_A(e) \geq 0$), and the second inequality follows from the weak submodularity of $f$ and $g$.
	We also have
	\begin{align}
	\label{eq2}
	h_A(B)
	& = f(A \cup B) g(A \cup B) - f(A) g(A) \nonumber \\ &
	= f(A \cup B) g(A \cup B) - f(A \cup B) g(A) + f(A \cup B) g(A) - f(A) g(A) \nonumber \\
	& = f(A \cup B) g_A(B) + f_A(B) g(A),
	\end{align}
	and
	\begin{align}
	\label{eq3}
	h_A(B)
	& = f(A \cup B) g(A \cup B) - f(A) g(A) \nonumber \\ &
	= f(A \cup B) g(A \cup B) - g(A \cup B) f(A) + g(A \cup B) f(A) - f(A) g(A) \nonumber \\
	& = g(A \cup B) f_A(B) + g_A(B) f(A).
	\end{align}
	Now assume that $\gamma^g_{A,B} \leq \gamma^f_{A,B}$. Then combining expressions~\eqref{eq1} and~\eqref{eq2} we obtain
	\begin{align*}
	\frac{f(A)}{f(A \cup B)} \cdot \gamma^g_{A,B} \cdot h_A(B) & =  f(A) \cdot \gamma^g_{A,B} \cdot g_A(B) + \frac{f(A)}{f(A \cup B)} \cdot \gamma^g_{A,B} \cdot f_A(B) \cdot g(A) \\
	& \leq f(A) \cdot \gamma^g_{A,B} \cdot g_A(B) + \gamma^f_{A,B}  \cdot f_A(B) \cdot g(A) \leq \sum_{e \in B}h_A(e),
	\end{align*}
	where the first inequality holds since $\frac{f(A)}{f(A \cup B)} \cdot \gamma^g_{A,B} \leq \gamma^g_{A,B} \leq \gamma^f_{A,B} $.

	In a similar fashion, assume that $\gamma^f_{A,B} \leq \gamma^g_{A,B} $. Then combining expressions~\eqref{eq1} and~\eqref{eq3} we obtain
	\begin{align*}
	\frac{g(A)}{g(A \cup B)} \cdot \gamma^f_{A,B} \cdot h_A(B) & =  g(A) \cdot \gamma^f_{A,B} \cdot f_A(B) + \frac{g(A)}{g(A \cup B)} \cdot \gamma^f_{A,B} \cdot g_A(B) \cdot f(A) \\
	& \leq g(A) \cdot \gamma^f_{A,B} \cdot f_A(B) + \gamma^g_{A,B}  \cdot g_A(B) \cdot f(A) \leq \sum_{e \in B}h_A(e),
	\end{align*}
	where the first inequality holds since $\frac{g(A)}{g(A \cup B)} \cdot \gamma^f_{A,B} \leq \gamma^f_{A,B} \leq \gamma^g_{A,B}$.
	This proves the first part of the statement.

	In a very similar fashion we can show the second part. Now assume that $f$ and $g$ have global parameters $\gamma^f$ and $\gamma^g$ respectively, such that $\gamma^f \geq \gamma^g$. Then we have $\gamma^g_{A,B} \geq \gamma^g$ and $\gamma^f_{A,B} \geq \gamma^f \geq \gamma^g$ for all pairs of disjoint sets $A$ and $B$. By using Equation~\eqref{eq1} we get
	\begin{align*}
	\sum_{e \in B} h_A(e)
	& \geq f(A) \cdot \gamma^g \cdot g_A(B) + \gamma^g \cdot f_A(B) \cdot g(A) = \gamma^g \cdot \big[f(A) \cdot  g_A(B) + f_A(B) \cdot g(A) \big].
	\end{align*}
	This leads to the bounds
	\begin{align*}
	\sum_{e \in B} h_A(e)
	& \geq \gamma^g \cdot \big[f(A) \cdot  g_A(B) + f_A(B) \cdot g(A) \big] \\ &
	\geq \gamma^g \cdot \left[f(A) \cdot  g_A(B) + \frac{f(A)}{f(A \cup B)} \cdot f_A(B) \cdot g(A) \right]
	= \frac{f(A)}{f(A \cup B)} \cdot \gamma^g \cdot h_A (B),
	\end{align*}
	and
	\begin{align*}
	\sum_{e \in B} h_A(e)
	& \geq \gamma^g \cdot \big[f(A) \cdot  g_A(B) + f_A(B) \cdot g(A) \big] \\ &
	\geq \gamma^g \cdot \left[ \frac{g(A)}{g(A \cup B)} \cdot f(A) \cdot  g_A(B) +  f_A(B) \cdot g(A) \right]
	= \frac{g(A)}{g(A \cup B)} \cdot \gamma^g \cdot h_A (B).
	\end{align*}
	Hence the product function $h$ has parameter $\gamma_{A,B} \geq  \gamma^g \cdot \max \{ \frac{f(A)}{f(A \cup B)} , \frac{g(A)}{g(A \cup B)} \}$. This concludes the proof.
\end{proof}

\prodcardisubmod*
\begin{proof}
	Let $A,B \subseteq E$ be two disjoint sets, and denote $a=|A|$ and $b=|B|$.
	We prove (a) first. Note that
	\[
	g_A(e) = (a+1)f(A+e)- a f(A) = a f_A(e) + f(A+e) = (a+1) f_A (e) + f(A).
	\]
	Hence
	\[
	\sum_{e \in B} g_A(e) = (a+1) \sum_{e \in B} f_A (e) + b f(A) \geq (a+1) f_A (B) + b f(A) ,
	\]
	where the inequality follows from submodularity. We also have
	\[
	g_A(B) = (a+b) f(A \cup B) - a f(A) = a f_A (B) + b f(A \cup B) = (a+b) f_A (B) + b f(A).
	\]
	Combining the above two expressions we immediately obtain
	\[
	\sum_{e \in B} g_A(e) \geq (a+1) f_A (B) + b f(A) \geq \frac{a+1}{a+b} g_A(B),
	\]
	where the last inequality uses the non-negativity of $f$, i.e., that $f(A)\geq 0$.

	We now proceed to prove statement (b).
	For this case observe that
	\[
	g_A(e) = \frac{f(A+e)}{a+1} - \frac{f(A)}{a} = \frac{af(A+e)- (a+1)f(A)}{a(a+1)} = \frac{af_A (e) - f(A)}{a(a+1)} = \frac{f_A (e) - g(A)}{a+1}.
	\]
	Hence
	\[
	\sum_{e \in B} g_A(e) = \sum_{e \in B}  \frac{f_A (e) - g(A)}{a+1} = \frac{\sum_{e \in B} f_A (e) - b g(A)}{a+1} \geq \frac{f_A (B) - b g(A)}{a+1},
	\]
	where the inequality follows from submodularity. We also have
	\[
	g_A(B) = \frac{f(A\cup B)}{a+b} - \frac{f(A)}{a} = \frac{af(A \cup B)- (a+b)f(A)}{a(a+b)} = \frac{af_A (B) - bf(A)}{a(a+b)} = \frac{f_A (B) - bg(A)}{a+b}.
	\]
	Combining the above two expressions we immediately obtain
	\[
	\sum_{e \in B} g_A(e) \geq \frac{f_A (B) - b g(A)}{a+1} = \frac{a+b}{a+1} g_A(B).
	\qedhere
	\]
\end{proof}

\sumproperty*
\begin{proof}
	Let  $A,B$ be two disjoint sets. 
	Then when $f,g$ are both monotone we have
	\begin{align*}
	\sum_{e \in B} {(f+g)}_{A}(e) &= \sum_{e \in B} f_{A} (e) + \sum_{e \in B} g_{A} (e) \geq
	\gamma^f_{A,B} \cdot f_{A} (B) + \gamma^g_{A,B} \cdot g_{A} (B) \\ &
	\geq \min\{ \gamma^f_{A,B}, \gamma^g_{A,B}\} \cdot [f_{A} (B) + g_{A} (B)] = \min\{ \gamma^f_{A,B}, \gamma^g_{A,B}\} \cdot {(f+g)}_{A} (B),
	\end{align*}
	where the second inequality follows since $f_A (B), g_A (B) \geq 0$ by monotonicity of $f$ and $g$. Similarly, in the case where $f$ is monotone and $g$ is non-monotone we get
	\begin{align*}
	\sum_{e \in B} {(f+g)}_{A} (e) &= \sum_{e \in B} f_{A} (e) + \sum_{e \in B} g_{A} (e) \geq
	\gamma^f_{A,B} \cdot f_{A} (B) + \gamma^g_{A,B} \cdot g_{A} (B) \\ &
	\geq  \gamma^g_{A,B} \cdot [f_{A} (B) + g_{A} (B)] =  \gamma^g_{A,B} \cdot {(f+g)}_{A} (B),
	\end{align*}
	where the second inequality follows by the assumption $ \gamma^g_{A,B} \leq \gamma^f_{A,B}$ and the fact that $f_A (B)\geq 0$ by monotonicity of $f$.
\end{proof}

\divsubmod*
\begin{proof}
	Let  $A,B$ be two disjoint sets with $|A|=a$ and $|B|=b$. The first statement follows from Proposition~\ref{prop:sum-property} and the fact that  $\gamma^{f}_{A,B} \geq 1 \geq \gamma^{g}_{A,B} \geq  \frac{a}{a+b-1}$, where the last inequality follows from Example~\ref{ex:diversity}.

	For the second statement, recall that in Example~\ref{ex:|S|f(S)} we showed that the function $h(S):=|S| \cdot g(S)$ has parameter $\gamma_{A,B}\geq \frac{a+1}{a+b}$. Using that submodular functions have parameter $\gamma_{A,B}\geq 1$ we immediately get
	\begin{align*}
	\sum_{e \in B} {(f+g)}_{A} (e) &= \sum_{e \in B} f_{A} (e) + \sum_{e \in B} g_{A} (e) \geq
	f_{A} (B) + \frac{a+1}{a+b} \cdot g_{A} (B) \\ &
	\geq  \frac{a+1}{a+b} \cdot [f_{A} (B) + g_{A} (B)] =  \frac{a+1}{a+b} \cdot {(f+g)}_{A} (B),
	\end{align*}
	where in the second inequality we use that $f$ is monotone.

	The third statement follows in a very similar fashion by using the result from Example~\ref{ex:proportionally-submod}.
\end{proof}

\section{Applications for the product of monotone set functions}\label{sec:appendix-product}

\begin{claim}\label{claim:submodular-product}
	Let $f:2^E \to \R_+$ be a monotone submodular function, and denote $S^* := \argmax_{|S| \leq k} f(S)$.
	Let $U \subseteq E$ be such that $f(U) \geq \alpha \cdot f(S^*)$.
	Then, for any pair of sets $A,B$ such that $A \supseteq U$ and $|B| \leq k$ it holds that
	\[
	\frac{f(A) }{f(A \cup B)} \geq \frac{\alpha }{1+ \alpha }.
	\]
\end{claim}
\begin{proof}
	\begin{align*}
		\frac{f(A) }{f(A \cup B)}  & = \frac{f(A) }{f(A) + f_A(B)} \geq \frac{f(A) }{f(A) + f(B)}
		\geq \frac{f(A) }{f(A) + f(S^*)} = \frac{1 }{1+ f(S^*) / f(A)} \\ &
		\geq \frac{1 }{1+ 1 / \alpha } = \frac{\alpha }{1+ \alpha },
	\end{align*}
	where the first inequality follows by submodularity, the second inequality follows since $|B| \leq k$ and hence $f(B) \leq f(S^*)$, and the last inequality holds since $f(A) \geq f(U) \geq \alpha \cdot f(S^*)$.
\end{proof}

\begin{claim}\label{claim:diversity-product}
	Let $f:2^E \to \R_+$ be a monotone metric diversity function (as discussed in Example~\ref{ex:diversity}), and denote $S^* := \argmax_{|S| \leq k} f(S)$.
	Let $U \subseteq E$ be such that $|U| \leq k$ and  $f(U) \geq \alpha \cdot f(S^*)$.
	Then, for any pair of sets $A,B$ of size at most $k$ such that $A \supseteq U$, it holds that
	\[
	\frac{f(A) }{f(A \cup B)} \geq \frac{\alpha }{5+ \alpha }.
	\]
\end{claim}
\begin{proof}
	Let $A,B \subseteq E$ be as above, and without loss of generality assume they are disjoint (if they are not then take $A$ and $B\setminus A$ instead).
	Let $d(A,B)$ denote the sum of distances between elements in $A$ and $B$. That is, $d(A,B) = \sum_{u \in A, v \in B} d(u,v)$. Then one can write $f(A \cup B) = f(A) + f(B) + d(A,B)$.

	Since $A$ and $B$ have cardinality at most $k$, we can partition $A$ into two disjoint sets $A_1$ and $A_2$ of size at most $k/2$, and similarly we can partition $B$ into two disjoint sets $B_1$ and $B_2$ of size at most $k/2$. It then follows that
	\[
	d(A,B) = \sum_{i,j \in \{1,2\}} d(A_i,B_j) \leq \sum_{i,j \in \{1,2\}} f(A_i \cup B_j) \leq 4 f(S^*),
	\]
	where the last inequality follows since $A_i \cup B_j$ are sets of size at most $k$.

	Combining this with the fact that $B$ is a set of  size at most $k$ we get
	\[
	f(A \cup B) = f(A) + f(B) + d(A,B) \leq f(A) + 5 f(S^*).
	\]
	Finally, using monotonicity and the fact that $A \supseteq U$ we obtain
	\[
	\frac{f(A) }{f(A \cup B)}    \geq \frac{f(A) }{f(A) + 5f(S^*)} = \frac{1 }{1+ 5 f(S^*) / f(A)}
	\geq \frac{1 }{1+ 5 / \alpha } = \frac{\alpha }{5+ \alpha },
	\]
\end{proof}

\end{document}